\documentclass{article}

\usepackage{fullpage}
\usepackage{amsthm,amsmath,amssymb,amsfonts,xspace,color,enumerate,graphicx,xfrac,mathrsfs}

\newtheorem{lemma}{Lemma}

\newtheorem{theorem}{Theorem}
\newtheorem{cor}{Corollary}

\usepackage{algorithm}
\usepackage{algorithmicx}
\usepackage[noend]{algpseudocode}

\makeatletter
\def\BState{\State\hskip-\ALG@thistlm}
\makeatother

\algnewcommand\algorithmicinput{\textbf{\ \ \ \ \ \ \ Input:}}
\algnewcommand\INPUT{\item[\algorithmicinput]}

\algnewcommand\algorithmicoutput{\textbf{\ \ \ \ \ \ \ Output:}}
\algnewcommand\OUTPUT{\item[\algorithmicoutput]}

\algnewcommand\algorithmicglobal{\textbf{\ \ \ \ \ \ \ Global Vars:}}
\algnewcommand\GLOBAL{\item[\algorithmicoutput]}

\title{Capacitated Dominating Set on Planar Graphs}
\author{Amariah Becker}

\begin{document}
\maketitle

\begin{abstract}
{\sc Capacitated Domination} generalizes the classic {\sc Dominating Set} problem by specifying for each vertex a required demand and an available capacity for covering demand in its closed neighborhood.  The objective is to find a minimum-sized set of vertices that can cover all of the graph's demand without exceeding any of the capacities. In this paper we look specifically at domination with hard-capacities, where the capacity and cost of a vertex can contribute to the solution at most once.  Previous complexity results suggest that this problem cannot be solved (or even closely approximated) efficiently in general.  In this paper we present a polynomial-time approximation scheme for {\sc Capacitated Domination} in unweighted planar graphs when the maximum capacity and maximum demand are bounded.  We also show how this result can be extended to the closely-related {\sc Capacitated Vertex Cover} problem.

\end{abstract}

\section{Introduction}\label{sec:intro}

The {\sc Dominating Set} problem is a classic NP-complete optimization problem defined as follows.  For an undirected graph $G$ with vertex set $V$, a dominating set $S$ is a subset of $V$ such that every vertex $v\in V$ is either in $S$ or adjacent to a vertex in $S$.  The problem seeks a dominating set such that $|S|$ is minimized.  The {\sc Dominating Set} problem has several applications in resource allocation~\cite{Fujita} and social network theory~\cite{Wang}.

An $\alpha$-approximation for a minimization problem ($\alpha \geq 1$) is a polynomial-time algorithm that returns a solution whose cost is within an $\alpha$ factor of the optimum value.  A polynomial-time approximation scheme (PTAS) exists for a problem if for any $\epsilon > 0$ there is a $(1+\epsilon)$-approximation.  An \emph{efficient} polynomial-time approximation scheme (EPTAS) is a PTAS in which for any $\epsilon > 0$ the $(1+\epsilon)$-approximation has an $O(f(\epsilon)n^c)$ runtime in which $c$ is a constant that does not depend on $\epsilon$.  Assuming ${\rm P} \neq {\rm NP}$, a PTAS is in some sense the best solution we can hope to achieve for NP-hard problems in polynomial time.  

For the {\sc Dominating Set} problem in general graphs, an $O(\ln n)$-approximation is known~\cite{johnson}.  It has also been shown that unless ${\rm NP} \subset {\rm TIME}(n^{\log\log n})$ the {\sc Dominating Set} problem cannot be approximated within a factor of $(1-\epsilon)\ln n$~\cite{feige}, thus a PTAS is unlikely to exist.  For restricted graph classes such as unit disk graphs~\cite{nieberg}, bounded-degree graphs~\cite{chlebik}, and planar graphs~\cite{marzban}, improvements on this bound have been found.  In particular, the planar case is known to admit a PTAS by extending the framework established by Baker~\cite{baker} and described below. 

For many optimization problems that have applications in resource allocation, a natural generalization is to consider capacity restrictions on the resources.  For covering-type problems, for example, we may impose capacity limits on the number of times an edge or vertex can be included in the solution or on how much of a commodity can be stored at or move through a location.  Capacitated versions of many common optimization problems have been studied, including {\sc Facility Location}~\cite{melkote}, {\sc Vertex Cover}~\cite{guha}, and {\sc Vehicle Routing}~\cite{ralphs}.  For the {\sc Dominating Set} problem, the capacitated version restricts the \emph{capacity} of vertices (i.e. the number of vertices that can claim a given vertex as their representative in the dominating set).

The capacitated {\sc Dominating Set} problem, often referred to as the {\sc Capacitated Domination} problem (CDP) assigns to each vertex a \emph{demand} value, a \emph{capacity} value, and a \emph{weight} value and seeks a dominating set $S$ of minimum weight such that every vertex in $V$ has its demand \emph{covered} by vertices in $S$, the total amount of demand covered by any vertex in $S$ does not exceed its capacity, and any vertex can only cover itself or its neighbors (i.e. the \emph{closed neighborhood})\footnote{In one common variant of the problem, a vertex can cover its own demand `for free' and the capacity only limits coverage of neighbors.}.  We use $CDP(G,d,c)$ to denote the instance of the {\sc Capacitated Domination} problem on a graph $G$ with demand function $d$ and capacity function $c$.

CDP was shown to be ${\rm W}[1]$-hard for general graphs when parameterized by treewidth~\cite{dom}.  Moreover, this hardness result extends to planar graphs, and even the bidimensionality framework does not apply to this problem~\cite{bodlaender}.

Kao et al. characterize CDP along several axes~\cite{kao}.  The \emph{weighted} version has real, non-negative vertex weights whereas the \emph{unweighted} version has uniform vertex weights of weight one.  The \emph{unsplittable}-demand version requires that the entire demand of any vertex is covered by a single vertex in the dominating set, while the \emph{splittable}-demand version allows demand to be covered by multiple vertices in the dominating set.  The \emph{soft}-capacity version allows multiple copies of vertices to be included in the dominating set (effectively scaling both the capacity and the weight of vertex by any integral amount) whereas the \emph{hard}-capacity version limits the number of copies of each vertex (usually to one).  Lastly, we can specify whether or not the capacities and demands are required to be (non-negative) integral.

For the soft-capacity version, many results are known.  For general graphs, Kao et al. present a $(\ln n)$-approximation for the weighted unsplittable-demand case, a $(4 \ln n + 2)$-approximation for the weighted splittable-demand case, a $(2 \ln n + 1)$-approximation for the unweighted splittable-demand case, and a $(\Delta)$-approximation for the weighted splittable-demand case, where $\Delta$ is the maximum degree in the graph~\cite{kao, kao2}.  As these are generalizations of the classic {\sc Dominating Set} problem, there is little room for improvement for these approximations on general graphs.  

For the soft-capacity problem in planar graphs, Kao et al. developed a constant-factor approximation algorithm~\cite{kao3}.  Additionally, they present a pseudo polynomial-time approximation scheme for the weighted splittable-demand case by first showing how to solve the problem on graphs of bounded treewidth and then applying Baker's framework~\cite{kao2}.

Substantially fewer results have been published on CDP with hard-capacities.  In this paper we consider the unweighted, splittable-demand, hard-capacity version of CDP with integral capacities and demands.  Our main result is stated in Theorem~\ref{thm:main1}.  

\begin{theorem}\label{thm:main1} For any positive integers $d^*$ and $c^*$ there exists a PTAS for the {\sc Capacitated Domination} problem, $CDP(G,d,c)$ where $G$ is a planar graph, $d$ is a demand function with maximum value $d^*$, and $c$ is a capacity function with maximum value $c^*$.\end{theorem}

This result for hard-capacities mirrors the PTAS presented by Kao et al. for soft-capacities~\cite{kao2}.  The challenge in extending the Baker framework to the hard-capacitated version is that the technique is agnostic to vertices being duplicated in the solution and does not necessarily return a feasible solution when we disallow duplicates (or higher multiplicities).   Our main contributions are twofold.  First we show how to \emph{smooth} a capacity-exceeding solution into a capacity-respecting solution.  Second we modify the traditional Baker framework to bound the cost of this smoothing process.

In Section~\ref{sec:prelim} we introduce the problem and give an overview of Baker's framework.  Section~\ref{sec:ptas} presents the PTAS and gives a proof of our main theorem and several corollaries including an extension to {\sc Capacitated Vertex Cover}.  In Section~\ref{sec:dp} we give the dynamic programming algorithm for solving CDP exactly for graphs of bounded branchwidth.  Finally, in Section~\ref{sec:conclude} we discuss possible extensions of the techniques presented.


\section{Preliminaries}\label{sec:prelim}

Let $G = (V,E)$ be an undirected graph with vertex set $V$ and edge set $E$.  For a subset of edges, $F\subseteq E$, we define $\partial_G(F)$ to be the set of vertices in $V$ incident both to edges in $F$ and $E \setminus F$.  For a subset of vertices, $S\subseteq V$, we define $\delta_G(S)$ to be the set of edges in $E$ incident both to vertices in $S$ and $V \setminus S$ and $\partial_G(S)$ to be the set of vertices in $S$ incident to edges in $\delta_G(S)$.  When $G$ in unambiguous, we exclude it from the subscript.

We consider a demand function $d: V \rightarrow \mathbb{Z}^*$ and a capacity function $c: V \rightarrow \mathbb{Z}^*$.  The capacity (resp. demand) of a set of vertices is the sum of the constituent vertex capacities (resp. demands).  Namely, for $R \subseteq V$, $c(R) = \sum_{v \in R} c(v)$ and $d(R) = \sum_{v \in R}d(v)$.

An \emph{assignment} $\mathcal{A}$ is a multiset of ordered pairs $(u,v)$ with $u,v \in V$ such that either $(u,v) \in E$ or $u = v$, and $\forall v \in V,\ \left\vert\{(x,v)\in \mathcal{A}\ :\ x\in V\}\right\vert \leq d(v)$.  An assignment is \emph{proper} if $\forall u \in V,\ \left\vert\{(u,x) \in \mathcal{A}\ :\ x\in V\}\right\vert \leq c(u)$. An assignment is \emph{covering} if $\forall v \in V,\ \left\vert\{(x,v)\in \mathcal{A}\ :\ x\in V\}\right\vert = d(v)$.  The total unmet demand for an assignment is $t(\mathcal{A}) = \sum_{v\in V}{(d(v) - \left\vert\{(x,v)\in \mathcal{A}\ :\ x\in V\}\right\vert)}$.  Observe that if $\mathcal{A}$ is covering then $t(\mathcal{A}) = 0$. A given assignment $\mathcal{A}$ has \emph{dominating set}
$S(\mathcal{A}) = \{u\in V \ :\  (u,v) \in \mathcal{A}\}$ and \emph{size} $q(\mathcal{A}) = \left\vert S(\mathcal{A})\right\vert$.

The objective of the {\sc Capacitated Domination} problem (CDP) is to find a proper, covering assignment $\mathcal{A}$ such that $q(\mathcal{A})$ is minimized. 

Conceptually, in a solution to CDP, $\mathcal{A}$ describes an assignment of vertices to neighbors such that the demand of every vertex is met and the capacity of every vertex is not exceeded.  This generalizes the classic {\sc Dominating Set} problem in which every vertex $v$ has $d(v) = 1$ and $c(v) = \infty$.  We can think of the ordered pairs in the assignment as representing (facility, client) pairs.  

We use superscripts to denote the multiplicity of an ordered pair in an assignment.  For example $\{($u$,$v$)^4,($w$,$w$)\}$ assigns four units of demand from the vertex $v$ to the vertex $u$ and one unit of demand from $w$ to itself.

We use $\mathcal{A}_1 \uplus \mathcal{A}_2$ to denote the multiset union operation in which multiplicities are additive.  For example $\{($u$,$v$)^2,($u$,$w$)\} \uplus \{($u$,$v$)^3,($w$,$v$)\} = \{($u$,$v$)^5,($u$,$w$),($w$,$v$)\}$.

A \emph{planar} graph is one that can be embedded on a sphere without any edge-crossings.

A \emph{branch-decomposition} of $G$ is a rooted binary tree $T$ and a bijective mapping of edges of $E$ to leaves of $T$.  The nodes of $T$ map to the subsets of $E$ (called \emph{clusters}) consisting of the corresponding edges of all descendant leaves.  The root cluster therefore contains all edges in $E$, and leaf clusters are singleton edge sets. The \emph{width} of a branch-decomposition is $\max_C|\partial(C)|$ over all clusters $C$.  The \emph{branchwidth} of $G$ is the minimum width over all branch-decompositions of $G$.  For planar graphs, an optimal branch decomposition can be computed in polynomial time~\cite{gu}.
The hierarchical structure of branch-decompositions lends itself well to dynamic-programming techniques.  For many graph optimization problems, given a branch-decomposition with width $w$ for the input graph, an optimal solution can be found in $2^{O(w)}n$ time~\cite{dorn}. We present such a dynamic program in Section~\ref{sec:dp}.

The PTAS presented in Section~\ref{sec:ptas} is based on Baker's framework for designing PTASs in planar graphs ~\cite{baker}. The framework decomposes the input graph into subgraphs of bounded branchwidth, each of which can be solved efficiently, then forms an overall solution by combining the subgraph solutions.  The error incurred can be charged to the interface between these subgraphs, and a shifting argument can be used to bound this error.

Specifically, Baker's approach has been applied to {\sc Dominating Set} (and soft-{\sc Capacitated Domination}) on planar graphs as follows~\cite{marzban, kao2}.  The input graph is partitioned into \emph{levels} (vertex subsets) defined by minimum hop-distance (breadth-first search) from an arbitrary root $r$.  The resulting subgraph induced by $k$ consecutive levels has $O(k)$ branchwidth.  We decompose the graph into subgraphs of $k$ consecutive levels such that every two adjacent subgraphs overlap in two \emph{boundary} levels.  There are at most $k$ different \emph{shifts} that can divide the graph into $k$-level subgraphs in this way.  For each shift, we solve the (capacitated) {\sc Dominating Set} problem on each (bounded-branchwidth) subgraph with the slight modification that the vertices of the topmost and bottommost levels of each subgraph require zero demand (but can still be part of the dominating set).  This can be done in $2^{O(w)}n$ time.  The overall solution is then the union of the solutions for the subgraphs.  

Ignoring the demand of the outermost levels prevents a doubling-up of the demand requirements of the vertices in the boundary levels (adjacent subgraph overlaps) while still allowing these vertices to satisfy demand from inner levels in the subgraph. Taking the union ensures that any ignored demand is met by some subgraph (no vertex's demand is ignored by multiple subgraphs).  Since each subgraph is solved optimally and an optimum solution, $X$, for the entire graph induces solutions on each subgraph, the sum of the costs of the subgraph solutions is bounded by the cost of $X$ plus some amount of error that can be bounded by the weight of $X$ that intersects the boundary levels.  For some \emph{shift} this weight is a small fraction ($O(1/k)$) of the total weight of the optimum solution.

The problem with extending this approach to hard-capacities emerges when a vertex in a boundary level is included in the solutions for two different subgraphs. For the classic {\sc Dominating Set} problem, these vertices have unbounded capacities so taking the union of solutions is not an issue.  Similarly, for the soft-{\sc Capacitated Domination} problem, we are allowed to include multiple copies of these vertices, and since the weight is small we can afford to do so.  However, for the hard-{\sc Capacitated Domination} problem these vertices pose a problem since we may have \emph{overloaded} them (exceeded their capacities) and are not allowed to simply pay the extra cost to use these vertices multiple times.

In this paper, we address this issue by redirecting the assignment of overloaded vertices to underloaded vertices elsewhere in the graph.  We may have to do this for every solution vertex that appears on the boundary levels.  Unfortunately the above bound on the error arises from the sum of subgraph solution costs and does not bound the weight of the solution vertices that appear on boundary levels.  We address this problem by modifying the way that the graph is decomposed into subgraphs: instead of every subgraph being constructed and used identically (as is typical in the Baker framework), we alternate between two different types of subgraphs.  This facilitates accounting for the reassignment of overloaded vertices.


\section{PTAS}\label{sec:ptas}

The PTAS for {\sc Capacitated Domination} uses the approach introduced by Baker~\cite{baker} and described in Section~\ref{sec:prelim}. Breadth-first search from an arbitrary root vertex, $r$, partitions the graph into $m$ levels (defined by hop-distance from $r$) such that the subgraph induced by any $k$ consecutive levels has branchwidth at most $2k$.  

For a given $i,j,k,l$ such that $0\leq i < k$, $0\leq j \leq \lfloor\frac{m}{k}\rfloor$, and $0\leq l \leq \lfloor\frac{m}{k}\rfloor -1$, let $G^i_j$ denote the subgraph induced by levels $jk+i$ through $(j+1)k+i-1$ and $H^i_l$ denote the subgraph induced by levels $(l+1)k+i-2$ through $(l+1)k+i + 1$.  For a fixed $k$ and $i$ we call $G^i_j$ the $j\text{th}$ \emph{slab} and $H^i_l$ the $l\text{th}$ \emph{patch}.  We call $k$ the slab \emph{height} and $i$ the slab \emph{shift}.

Consecutive slabs do not overlap each other and each slab has branchwidth at most $2k$.  Each patch occurs at the interface between two consecutive slabs, overlapping with two levels from each, and has constant branchwidth (at most 8).

The algorithm proceeds as follows (see pseudocode).  Independently solve (a slightly modified version of) CDP restricted to each slab and each patch.  Then take the union of solutions over all slabs and patches.  Note that since patches overlap the slabs, if a dominating  vertex from a proper patch solution is also a dominating vertex from a proper slab solution, the union of solutions may no longer be proper.  In fact the capacities for patch vertices have been effectively doubled.  The algorithm alleviates this excess strain on patch vertices by reassigning the demand to \emph{underfull} vertices elsewhere in the graph.

Specifically, for each slab $G^i_j$ and each patch $H^i_l$ modify the problem so that the vertices of the topmost and bottommost levels have no demand (exclude the top level of the first slab and the bottom level of the last slab).  Since patches overlap the two bottommost and two topmost levels of slabs, any demand removed from a slab will be covered by a patch and vice versa. The branchwidth of the slab is at most $2k$, so we can use the dynamic programming algorithm $\Call{Cdp\_dp}{}$ described in Section~\ref{sec:dp} to generate a solution to the modified problem on the slab in polynomial time.  

Let $\mathcal{A}_{G^i_j}$ denote the optimal assignment for the modified CDP for slab $G^i_j$ and $\mathcal{A}_{H^i_l}$ denote the optimal assignment for patch $H^i_l$.  Combining slab and patch assignments gives $\tilde{\mathcal{A}_i} = \Big(\biguplus_{j} \mathcal{A}_{G^i_j}\Big) \uplus \Big(\biguplus_l \mathcal{A}_{H^i_l}\Big)$.  Since $\tilde{\mathcal{A}_i}$ may no longer be a proper assignment, we use a \emph{smoothing} procedure, $\Call{Smooth}{}$, to generate a final solution $\mathcal{A}'_i$. $\Call{Smooth}{}$ first removes ordered pairs from an assignment $\mathcal{A}$ until it is proper by repeatedly finding an \emph{overloaded} vertex $v$ and removing a pair $(v,u)$ from $\mathcal{A}$ until no overloaded vertex can be found. $\Call{Smooth}{}$ then reassigns any remaining unmet demand to \emph{underfull} vertices in the graph with excess capacity by repeatedly finding alternating paths in the induced assignment graph (described in Lemma~\ref{lem:improve}).  We use the following lemmas to show that such a reassignment is feasible.  Finally, the algorithm returns the assignment $\mathcal{A}'_i$ with the minimum cost among all choices of shift $i$. 

\begin{algorithm}
   \caption{Capacitated Domination Problem PTAS}
    \begin{algorithmic}[1]
      \Function{CDP\_PTAS}{$G_0, c, d, k$}
      \INPUT{Graph $G_0$, capacity function $c$, demand function $d$, positive integer $k$}
      \OUTPUT{Minimum proper, covering assignment $\mathcal{A}$}
      \For{$i$ in 0...$k-1$}
			\State{$\tilde{\mathcal{A}_i} \gets \emptyset$}
   	   \For{$j$ in 0...$\lfloor m/k\rfloor$}
	   		\State{$\tilde{\mathcal{A}_i} \gets \tilde{\mathcal{A}_i} \uplus \Call{CDP\_DP}{G^i_j, c, d}$}
			\EndFor
		   \For{$l$ in 0...$\lfloor m/k\rfloor$-1}
		   	\State{$\tilde{\mathcal{A}_i} \gets \tilde{\mathcal{A}_i} \uplus \Call{CDP\_DP}{H^i_l, c, d}$}
			\EndFor
			\State{$\mathcal{A}'_i \gets \Call{Smooth}{G_0, \tilde{\mathcal{A}_i}, c, d}$}
		\EndFor
      \State\Return $\min_i\mathcal{A}'_i$
\EndFunction
\end{algorithmic}
\end{algorithm}


\begin{lemma}\label{lem:cut} Given a graph $G$ with demand and capacity functions $d$ and $c$ and a set of vertices $R \subseteq V$, if a proper, covering assignment exists for the {\sc Capacitated Domination} problem then $c(R) + d(\partial(R)) \geq d(R) $.
\end{lemma}

\begin{proof}
Assume that a proper, covering assignment $\mathcal{A}$ exists.  For any subset of vertices $R \subseteq V$ the total demand, $d(R)$, of vertices in $R$ must be met by vertices in the dominating set $S(\mathcal{A})$.  The vertices in $R\cap S(\mathcal{A})$ can contribute at most $c(R\cap S(\mathcal{A})) \leq c(R)$ toward meeting demand from $R$.  Vertices in $V \setminus R$ can only cover demand from $R$ via adjacent vertices, namely those in $\partial(R)$, so vertices in $(V \setminus R)\cap S(\mathcal{A})$ can meet at most $d(\partial(R))$ of the total demand from $R$. But $S(\mathcal{A}) = (R\cap S(\mathcal{A})) \cup ((V \setminus R)\cap S(\mathcal{A}))$ can meet the entire demand from $R$.  This gives $c(R) + d(\partial(R)) \geq c(R\cap S(\mathcal{A})) + d(\partial(R)) \geq d(R)$.
\end{proof}


\begin{lemma}\label{lem:improve} {Given a graph $G = (V,E)$ with demand and capacity functions $d$ and $c$ and a proper assignment $\mathcal{A}'$ with unmet demand $t(\mathcal{A}') > 0$, if a proper, covering assignment exists then a new proper assignment $\mathcal{A}''$ can be found in polynomial time such that $q(\mathcal{A}'') \leq q(\mathcal{A}') + 1$ and $t(\mathcal{A}'') < t(\mathcal{A}')$}.
\end{lemma}

\begin{proof} Let $U$ be the set of vertices with unmet demand.  We assume that these vertices are at full capacity, otherwise simply assign such an underfull vertex, $v$, to itself and satisfy the lemma by letting $\mathcal{A}'' = \mathcal{A}' \uplus \{(v,v)\}$.  We further assume that a vertex $v \in U$ spends its entire capacity toward meeting its own demand, so $(v,v)$ occurs $c(v)$ times in $\mathcal{A}'$. We can always reassign $v$'s capacity in this way without increasing overall unmet demand or the size of $\mathcal{A}'$.

Use assignment $\mathcal{A}'$ to define an arc set on $V$ such that $(u,v) \in \mathcal{A}'$ indicates an arc from $u$ to $v$, and for every vertex $u$, the outdegree $\Delta^-(u)$ is at most $c(u)$.  The dominating set $S(\mathcal{A}')$ is $\{v\ :\ \Delta^-(v) > 0\}$, and a covering assignment is one in which, for every vertex $u$, the indegree $\Delta^+(u)$ equals $d(u)$.

Consider the directed \emph{assignment graph} $\mathscr{A}$ defined by this arc set.  We can assume that $\mathscr{A}$ has no directed cycles (except for self-loops), otherwise every arc in such a cycle can be changed to a self-loop without changing the total unmet demand or the size of $\mathcal{A}'$.

If a proper, covering assignment exists then there must be a set $W$ of vertices that are not at full capacity.

Let $G' = (V',E')$ denote $G = (V,E)$ augmented with self-loops at every vertex. Formally, $V'=V$ and $E' = E \cup \{vv\ :\ v\in V\}$.  We define a \emph{semi-alternating} directed path $P = p_0,p_1,...,p_l$ to be a path in $G'$ such that $\forall i$ $p_ip_{i+1} \in E'$ and $(p_i,p_{i+1}) \notin \mathcal{A}' \Leftrightarrow (p_{i+1},p_{i+2}) \in \mathcal{A}'$.  That is, $P$ is a path through $G'$ in which exactly every other edge is a \emph{forward} arc in the assignment graph $\mathscr{A}$.  We show that there exists such a \emph{semi-alternating} directed path $P = p_0,p_1,...,p_l$ from a vertex $p_0 \in U$ to a vertex $p_l \in W$.

Assume to the contrary that no such path exists.  Let $U^+$ be the set of all vertices $v \in V$ such that there exists a semi-alternating directed path from $u$ to $v$ for some $u \in U$.  Clearly $U \subseteq U^+$, and by assumption $U^+ \cap W = \emptyset$, so all vertices of $U^+$ are at full capacity.  Observe that $U \cap \partial_{G'}(U^+) = \emptyset$ and there are no self-loops in $\mathscr{A}$ at vertices in $\partial_{G'}(U^+)$ otherwise $U^+$ could be extended to include neighbors of such vertices.  For the same reason, any arcs of $\mathscr{A}$ into vertices in $\partial_{G'}(U^+)$ must start in $V'\setminus U^+$.  Since all of the demand of every vertex in $\partial_{G'}(U^+)$ is met and all arcs of $\mathscr{A}$ into $\partial_{G'}(U^+)$ start in $V'\setminus U^+$, we infer that $V'\setminus U^+$ contributes a total of $d(\partial_{G'}(U^+))$ toward meeting the demand of $U^+$. Furthermore all vertices in $U^+$ are at full capacity and there are no outgoing arcs of $\mathscr{A}$ in $\delta_{G'}(U^+)$ so $U^+$ contributes a total of $c(U^+)$ toward meeting its own total demand.  Yet $U \subseteq U^+$ so $U^+$ contains unmet demand.  So $c(U^+) + d(\partial_{G'}(U^+)) < d(U^+) $ which by Lemma~\ref{lem:cut} contradicts our feasibility assumption.

Such a path can be found in linear time using breadth-first search from vertices with insufficient demand met.

Given a \emph{semi-alternating} directed path $P = p_0,p_1,...,p_l$ from a vertex $p_0 \in U$ to a vertex $p_l \in W$, we know that $(p_0,p_1)$ is not in $\mathcal{A}'$ since $p_0$ is in $U$ and does not contribute to meeting the demand of its neighbors. We can also assume that $(p_{l-1},p_l)$ is not in $\mathcal{A}'$ since if the path ends in an assignment arc we can append a self-loop.  Reassign the arcs along this path as follows: 
$$\mathcal{A}'' = (\mathcal{A}' \setminus \{(p_1,p_2),(p_3,p_4),...,(p_{l-2},p_{l-1})\}) \uplus \{(p_1,p_0),(p_3,p_2),...,(p_{l-2},p_{l-3}),(p_l,p_{l-1})\}$$  
This gives $q(\mathcal{A}'') \leq q(\mathcal{A}') + 1$.  All of the demand previously met along path $P$ is still met although potentially by a different source, but $\mathcal{A}''$ also meets one additional unit of previously unmet demand at vertex $p_0$.  Therefore the resulting proper assignment $\mathcal{A}''$ satisfies $t(A'') < t(A')$.
\end{proof}

We restate the main theorem here for convenience.


\addtocounter{theorem}{-1}
\begin{theorem}\label{thm:main2} For any positive integers $d^*$ and $c^*$ there exists a PTAS for the {\sc Capacitated Domination} problem $CDP(G,d,c)$ where $G$ is a planar graph, $d$ is a demand function with maximum value $d^*$, and $c$ is a capacity function with maximum value $c^*$.\end{theorem}

\begin{proof}
Let $\mathcal{A}^*$ denote an optimal assignment for CDP and let $OPT = q(\mathcal{A}^*)$.  For a fixed shift $i$ let $\mathcal{A}^*_{G^i_j} = \mathcal{A}^* \cap G^i_j$ be the intersection of the optimal assignment with the $j^\text{th}$ slab and $OPT_{G^i_j}$ denote the size of the optimal assignment for the modified problem on the $j^\text{th}$ slab.  Similarly, let $\mathcal{A}^*_{H^i_l} = \mathcal{A}^* \cap H^i_l$ be the intersection of the optimal assignment with the $l^\text{th}$ patch and $OPT_{H^i_l}$ denote the size of the optimal assignment for the modified problem on the $l^\text{th}$ patch.

Because we modify the problem to ignore demand from the outermost levels of a slab (resp. patch) $\mathcal{A}^*_{G^i_j}$ (resp. $\mathcal{A}^*_{H^i_l}$) is a covering assignment for the modified problem on $G^i_j$ (resp. $H^i_l$) since all of the demand on the non-outermost levels will be covered by $\mathcal{A}^*_{G^i_j}$ (resp. $\mathcal{A}^*_{H^i_l}$).  Therefore $OPT_{G^i_j} \leq q(A^*_{G^i_j})$ and $OPT_{H^i_l} \leq q(\mathcal{A}^*_{H^i_l})$.

Our algorithm first determines the union, $\tilde{\mathcal{A}}$, of assignments on slabs and patches.  The size of this union is at most the sum of the sizes of the assignments for each slab and each patch.  We then remove and reassign at most $c^*$ assignment pairs for each duplicated vertex in the union (where $c^*$ is the value of the maximum capacity).  Duplicated vertices can only occur within the levels of the patches.  By Lemma~\ref{lem:improve} each reassignment adds at most one vertex to the dominating set.  Let $\mathcal{A}'$ be the final assignment output by the algorithm. Therefore
 $$q(\mathcal{A}') \leq \sum_j{OPT_{G^i_j}} + c^* \sum_l{OPT_{H^i_l}} \leq \sum_j{q(\mathcal{A}^*_{G^i_j})} + c^* \sum_l{q(\mathcal{A}^*_{H^i_l})}\leq OPT + c^* \sum_l{q(\mathcal{A}^*_{H^i_l})}$$
The final inequality comes from the slabs (and therefore their intersections with $\mathcal{A}^*$) being disjoint.  Let $i^*$ be the shift that minimizes the sum, $\sum_l{q(\mathcal{A}^*_{H^i_l})}$, of the intersection of patches with the optimal dominating set $S(\mathcal{A}^*)$.  

If $\sum_l{q(\mathcal{A}^*_{H^{i^*}_l})}$ exceeds $(4/k)OPT$ then the sum of patch intersections over all shifts, $\sum_i\sum_l{q(\mathcal{A}^*_{H^i_l})}$ exceeds $4OPT$ since $i^*$ was chosen from $k$ possible shifts to give the minimum such sum.  But each vertex in $\mathcal{A}^*$ appears in a patch for at most four different shifts so can contribute at most 4 to the sum, giving $\sum_i\sum_l{q(\mathcal{A}^*_{H^i_l})} \leq 4OPT$.  Therefore, $\sum_l{q(\mathcal{A}^*_{H^{i^*}_l})} \leq (4/k)OPT$.  This gives
 $$q(\mathcal{A}') \leq OPT + (4c^*/k)OPT$$
Setting $k$ to $4c^*/\epsilon$ gives a $(1+\epsilon)$-approximation.

The slabs can be solved in $O(2k{c^*}^{6k}{d^*}^{4k}n)$ time and patches can be solved in $O(8{c^*}^{24}{d^*}^{16}n)$ time using the dynamic program algorithm described in Section~\ref{sec:dp}. The algorithm then performs at most $\epsilon n$ reassignments, each of which can be computed in linear time.  This process is repeated $k$ times, giving an overall $O(k^2(c^*d^*)^{6k}n + c^*n^2)$ runtime.  The algorithm is therefore an EPTAS when we take $k$ to be constant.
\end{proof}

In fact, the capacity of a vertex never need exceed the sum of the demands in its closed neighborhood. If it does, we can remove the excess capacity without affecting the solution.  Conversely, the demand of a vertex can never exceed the sum of the capacities in its closed neighborhood. If it does, there is no feasible solution.  Therefore, we can assume $c^* \leq d^*(\Delta^* + 1)$ and $d^* \leq c^*(\Delta^* + 1)$, where $\Delta^*$ is the largest degree.  Our PTAS can thus be expressed in terms of any two of the three parameters $c^*, d^*, \Delta^*$, giving the following two corollaries.

\begin{cor} For any positive integers $d^*$ and $\Delta^*$ there exists a PTAS for the {\sc Capacitated Domination} problem, $CDP(G,d,c)$ where $G$ is a planar graph with maximum degree $\Delta^*$, $d$ is a demand function with maximum value $d^*$, and $c$ is a capacity function.
\end{cor}
\begin{cor} For any positive integers $c^*$ and $\Delta^*$ there exists a PTAS for the {\sc Capacitated Domination} problem, $CDP(G,d,c)$ where $G$ is a planar graph with maximum degree $\Delta^*$, $d$ is a demand function, and $c$ is a capacity function with maximum value $c^*$.
\end{cor}

Another corollary of the theorem is a PTAS for the {\sc Capacitated Vertex Cover} problem (CVCP), defined as follows.  Let $G$ be an undirected graph with edge set $E$ and vertex set $V$.  We define a capacity function $c: V \rightarrow \mathbb{Z}^*$ on the vertices of $G$ and a demand function $d: E \rightarrow \mathbb{Z}^*$ on the edges of $G$.  Under the constraint that edge-demand can only be met by its own endpoints, the objective is to find a set of vertices $S$ of minimum size such that every edge has its demand met by vertices in $S$ and each vertex in $S$ meets a total amount of demand that does not exceed its capacity.  We use $CVCP(G,d,c)$ to denote the {\sc Capacitated Vertex Cover} problem on a graph $G$ with demand function $d$ and capacity function $c$.

We can see that CVCP generalizes the traditional {\sc Vertex Cover} problem in which each vertex has unlimited capacity and each edge has a demand of one.  Furthermore, the {\sc Capacitated Vertex Cover} problem is identical to the {\sc Capacitated Domination} problem with the demand on edges instead of on vertices.  In the following corollary, we show that we can reduce CVCP to CDP to achieve a PTAS for the former.

\begin{cor}\label{thm:main} For any positive integers $d^*$ and $c^*$ there exists a PTAS for the {\sc Capacitated Vertex Cover} problem $CVCP(G,d,c)$ where $G$ is a planar graph, $d$ is a demand function with maximum value $d^*$, and $c$ is a capacity function with maximum value $c^*$.\end{cor}

\begin{proof}  Consider the following reduction from {\sc Capacitated Vertex Cover} to CDP.  Let $G'$ be the graph formed by bisecting every edge of $G$. Specifically, we define a new graph $G'$ with edge set $E'$, vertex set $V'$, demand function $d: V' \rightarrow \mathbb{Z}^*$, and capacity function $c: V' \rightarrow \mathbb{Z}^*$, as follows:

For every edge $e = uv \in E$, add edges $uw$ and $wv$ to $E'$ and add vertices $u,v,w$ to $V'$.  Additionally, set capacities to $c'(u) = c(u)$,  $c'(v) = c(v)$, and $c(w) = 0$, and set demands to $d'(w)= d(e)$ and $d'(u)=d'(v) = 0$.  This provides the necessary conditions to solve CDP on $G'$.

It is easy to see that a dominating set of $G'$ exactly corresponds to a vertex cover of $G$.  Setting the bisecting vertex to have zero capacity guarantees that these vertices will not be chosen in the minimum dominating set.  Moreover, maximum capacity $c^*$ and maximum demand $d^*$ are preserved.  Therefore the PTAS for CDP finds an appropriate vertex cover.
\end{proof}


\section{Dynamic Program}\label{sec:dp}
In this section we describe a dynamic programming algorithm for solving the {\sc Capacitated Domination} problem for graphs with bounded branchwidth.

For an input graph $G_0$ with branchwidth $w$, capacity function $c$, and demand function $d$, we use the branch decomposition of $G_0$ to guide the dynamic program.  Consider a cluster, $C$, in the decomposition.  We define functions $f, g: \partial(C) \rightarrow \mathbb{Z}$ such that $\forall v\in \partial(C)$, $0 \leq f(v) \leq c(v)$ and $0 \leq g(v) \leq d(v)$.  We use $f_\emptyset$ and $g_\emptyset$ to denote the special trivial functions for empty cluster boundaries.  We let $C_G$ denote the cluster containing all edges of $G$.  

Our dynamic programming table, $DP$, is indexed by triplets ($C,f,g$) for all such functions for each cluster.  The entry at index ($C,f,g$) is an optimal assignment $\mathcal{A}$ for a modified version of CDP for cluster $C$ such that for vertices at the boundary $\partial(C)$, $f$ and $g$ describe the \emph{exact} amount of capacity used and demand covered by the assignment.  We call this a \emph{restricted assignment}.  The algorithm proceeds level-by-level through the decomposition, starting at the leaves and continuing root-ward (see pseudocode).

We first address the leaves of the decomposition, which correspond to single edges in the graph.  For the leaf cluster of edge $uv$, we consider the following four cases.
\begin{itemize}
\item If $f(v) + f(u) \neq g(u) + g(v)$ there is no solution which we indicate by returning $\{(u,v)^\infty\}$.
\item Otherwise if $f(v) + f(u) = g(u) + g(v) = 0$ then we return $\emptyset$.
\item Otherwise if $f(v) \geq g(v)$ ($v$ has sufficient capacity to cover its own demand) then return\\ $\{(v,v)^{g(v)}, (v,u)^{f(v)-g(v)}, (u,u)^{f(u)}\}$.
\item Otherwise if $f(v) < g(v)$ ($v$ has insufficient capacity to cover its own demand) then return\\ $\{(u,u)^{g(u)}, (u,v)^{f(u)-g(u)}, (v,v)^{f(v)}\}$.
\end{itemize}

Moving root-ward, we show how to compute the restricted assignment for CDP for a parent cluster given the table entries for the child clusters.  Consider a cluster $C_0$ with child clusters $C_1$ and $C_2$.  We say that indices $(C_1, f_1, g_1)$ and $(C_2, f_2, g_2)$ are \emph{compatible} with $(C_0, f_0, g_0)$ if the following hold:

\begin{itemize}
\item $\forall v \in (\partial(C_1) \cap \partial(C_2)) \setminus \partial(C_0)$,  $f_1(v) + f_2(v) \leq c(v)$ and $g_1(v) + g_2(v) = d(v)$
\item $\forall v \in (\partial(C_0) \cap \partial(C_1)) \setminus \partial(C_2)$, $f_0(v)=f_1(v)$ and $g_0(v)= g_1(v)$
\item $\forall v \in (\partial(C_0) \cap \partial(C_2)) \setminus \partial(C_1)$, $f_0(v)=f_2(v)$ and $g_0(v)= g_2(v)$
\item $\forall v \in \partial(C_0) \cap \partial(C_1) \cap \partial(C_2)$, $f_1(v) + f_2(v) = f_0(v)$ and $g_1(v) + g_2(v) = g_0(v)$
\end{itemize}

To determine the entry in the table for index $(C_0, f_0, g_0)$ we search over all compatible pairs $(C_1, f_1, g_1)$ and $(C_2, f_2, g_2)$ with respective entries $\mathcal{A}_1$ and $\mathcal{A}_2$, and find the $\mathcal{A}_1 \uplus \mathcal{A}_2$ such that $q(\mathcal{A}_1 \uplus \mathcal{A}_2)$ is minimized.

Finally we return the assignment stored in $DP(C_G, f_\emptyset, g_\emptyset)$.

\begin{algorithm}
   \caption{Capacitated Domination Problem Dynamic Program}
    \begin{algorithmic}[1]
      \Function{CDP\_DP}{$G_0, c, d$}
      \INPUT{Graph $G_0$ with branchwidth $w$, capacity function $c$, and demand function $d$}
      \OUTPUT{Minimum covering, proper assignment $\mathcal{A}$}
      \State $DP \gets$  array initialized to $\infty$
      \For{$level$ in 0...height of branch decomposition}
      		\For{cluster $C$ in $level$}
      			\For{$f:\partial(C) \rightarrow \mathbb{Z}$, $g:\partial(C) \rightarrow \mathbb{Z}$ with $0 \leq f(v) \leq c(v)$ and  $0 \leq g(v) \leq d(v), \forall v$}
      				  \If{$C$ is a leaf containing edge $uv$}
      						\If{$f(v) + f(u) \neq g(u) + g(v)$}
	      						\State{$DP(C,f,g) \gets \{(u,v)^\infty\}$}
							\ElsIf{$f(v) + f(u) = g(u) + g(v) = 0$} 
								\State{$DP(C,f,g) \gets\emptyset$}
							\ElsIf{$f(v) \geq g(v)$} 
								\State{$DP(C,f,g) \gets \{(v,v)^{g(v)}, (v,u)^{f(v)-g(v)}, (u,u)^{f(u)}\}$}
							\Else
								\State{$DP(C,f,g) \gets \{(u,u)^{g(u)}, (u,v)^{f(u)-g(u)}, (v,v)^{f(v)}\}$}
					      \EndIf
					   \Else
					   		\For{pairs $(C_1,f_1,g_1)$ and $(C_2,f_2,g_2)$ compatible with $(C,f,g)$}
								\If{$q(DP(C_1,f_1,g_1)\uplus DP(C_2,f_2,g_2)) < DP(C,f,g)$}
									\State {$DP(C,f,g) \gets DP(C_1,f_1,g_1)\uplus DP(C_2,f_2,g_2)$}
								\EndIf
							\EndFor
					   \EndIf
		      	\EndFor
      		\EndFor
		\EndFor
      \State \Return $DP(C_G, f_\emptyset, g_\emptyset)$
\EndFunction
\end{algorithmic}
\end{algorithm}

\begin{theorem}\label{thm:dp}
For any positive integers $w$, $d^*$, and $c^*$ there exists an exact polynomial time algorithm for the {\sc Capacitated Domination} problem, $CDP(G,d,c)$ where $G$ is a planar graph with branch decomposition of width $w$, $d$ is a demand function with maximum value $d^*$, and $c$ is a capacity function with maximum value $d^*$.

\end{theorem}

\begin{proof}
We claim that the dynamic program described above is such an algorithm. We prove correctness using induction on the branch decomposition.

Our base cases are the leaves of the decomposition.  Let edge $uv$ be such a leaf.  We consider the same four cases addressed in the algorithm description.
\begin{itemize}
\item If $f(v) + f(u) \neq g(u) + g(v)$ the algorithm returns $\{(u,v)^\infty\}$ because the cluster contains a single edge so the total capacity used must equal the total demand covered for the cluster to have a feasible solution.
\item Otherwise if $f(v) + f(u) = g(u) + g(v) = 0$ the algorithm returns $\emptyset$ since there is no demand that needs to be met and no capacity available.
\item Otherwise if $f(v) \geq g(v)$ the algorithm returns $\{(v,v)^{g(v)}, (v,u)^{f(v)-g(v)}, (u,u)^{f(u)}\}$.  Since $f(v) + f(u) = g(u) + g(v)$, the demand from $v$ and $u$ is covered and the exact capacity is used. Note that if $f(v) = 0$ or $f(u) = 0$ the size of the dominating set for the returned assignment is one, otherwise the size is two.  Such an assignment is therefore a minimum.
\item Otherwise $f(v) < g(v)$, and the algorithm returns $\{(u,u)^{g(u)}, (u,v)^{f(u)-g(u)}, (v,v)^{f(v)}\}$. This case is symmetric to the previous one.
\end{itemize}

To show the inductive step, we consider a cluster $C_0$ and assume that the restricted assignment for CDP has been solved optimally by the dynamic program for its child clusters $C_1$ and $C_2$.  For index $(C_0, f_0, g_0)$ the algorithm searches over all compatible pairs $(C_1, f_1, g_1)$ and $(C_2, f_2, g_2)$ with respective entries $\mathcal{A}_1$ and $\mathcal{A}_2$ and returns the $\mathcal{A}_1 \uplus \mathcal{A}_2$ such that $q(\mathcal{A}_1 \uplus \mathcal{A}_2)$ is minimized.  Let $\mathcal{A}_1^* \uplus \mathcal{A}_2^*$ be such a minimum where $\mathcal{A}_1^*$ is the entry for $(C_1^*, f_1^*, g_1^*)$ and $\mathcal{A}_2^*$ is the entry for $(C_2^*, f_2^*, g_2^*)$.

To show that $\mathcal{A}_1^* \uplus \mathcal{A}_2^*$ is an optimal restricted assignment for $(C_0, f_0, g_0)$, we must show that boundary restrictions are met and that it is proper, covering, and minimum.  

By definition of compatible, boundary restrictions are automatically satisfied because every vertex on the boundary of $C_0$ must appear on the boundary of at least one of its children and compatibility forces agreement between the child (or sum of children for vertices on all three boundaries) and the parent.  

Since $\mathcal{A}_1^*$ and $\mathcal{A}_2^*$ are each proper, and compatibility requires that $f_1^*(v) + f_2^*(v)$ is at most $c(v)$ for each vertex $v$ on the shared boundary, $\mathcal{A}_1^* \uplus \mathcal{A}_2^*$ must be proper.
Similarly, since $\mathcal{A}_1^*$ and $\mathcal{A}_2^*$ are each covering, and compatibility requires that $g_1^*(v) + g_2^*(v)$ equals $d(v)$ for each vertex $v$ on the shared boundary, $\mathcal{A}_1^* \uplus \mathcal{A}_2^*$ must be covering as well.

Let $\mathcal{A}_0^*$ be a minimum restricted assignment for $(C_0, f_0, g_0)$.  Partition $\mathcal{A}_0^*$ into $X$ and $Y$ as follows:
\begin{itemize}
\item For $(v,v)\in \mathcal{A}_0^*$, with $v \in \partial{(C_1)} \cap \partial{(C_2)}$, put $(v,v)$ into either $X$ or $Y$ arbitrarily.
\item Otherwise for $(v,v)\in \mathcal{A}_0^*$, put $(v,v)$ in $X$ if $v$ is the endpoint of an edge in $C_1$ and in $Y$ if $v$ is the endpoint of an edge in $C_2$.
\item Otherwise for $(u,v)\in \mathcal{A}_0^*$, put $(u,v)$ in $X$ if edge $uv$ is in $C_1$ and in $Y$ if edge $uv$ is in $C_2$.
\end{itemize}

Clearly, every pair in assignment $\mathcal{A}_0^*$ must be in exactly one of $X$ and $Y$ so $\mathcal{A}_0^* = X \uplus Y$.

Let $f_x: \partial(C_1) \rightarrow \mathbb{Z}$ be the function that maps boundary vertices of $C_1$ to the capacity utilized by these vertices in $X$, and $g_x: \partial(C_1) \rightarrow \mathbb{Z}$ be the function that maps boundary vertices of $C_1$ to the demand from these vertices covered in $X$.  By construction, $X$ is now a restricted assignment for $(C_1, f_x, g_x)$.  Our inductive hypothesis gives $q(DP(C_1, f_x, g_x)) \leq q(X)$. Symmetrically, defining $f_y$ and $g_y$ similarly for $\partial(C_2)$ gives $q(DP(C_2, f_y, g_y)) \leq q(Y)$.  

$DP(C_1, f_x, g_x)$ and $X$ agree on the dominating vertices on the boundaries, as do $DP(C_2, f_y, g_y)$ and $Y$, so the above inequalities ensure
$$q(DP(C_1, f_x, g_x) \uplus DP(C_2, f_y, g_y)) \leq q(X \uplus Y) = q(\mathcal{A}_0^*)$$. 
Since $(C_1, f_x, g_x)$ and $(C_2, f_y, g_y)$ are compatible with $(C_0, f_0, g_0)$, $DP(C_1, f_x, g_x) \uplus DP(C_2, f_y, g_y)$ will be considered by the algorithm when finding the minimum such union. Therefore $q(\mathcal{A}_1^* \uplus \mathcal{A}_2^*) \leq q(DP(C_1, f_x, g_x) \uplus DP(C_2, f_y, g_y)) \leq q(\mathcal{A}_0^*)$ and is thus minimum.

Since the branch decomposition gives a binary tree with $O(n)$ leaves, there are $O(n)$ clusters.  For each cluster, $C$, the boundary $\partial(C)$ has at most $w$ vertices.  For each vertex $v \in \partial(C)$ there are $c(v)$ options for the value of $f$ and $d(v)$ options for the value of $g$.  There are thus at most ${c^*}^w$ functions $f$ and ${d^*}^w$ functions $g$ for each cluster boundary.  This gives $O((c^*d^*)^w n)$ indices of the DP table. Each entry requires comparing unions of entries from all compatible pairs. There are $O({c^*}^{2w}{d^*}^w)$ such pairs since knowing $f$ for one child and the parent completely determines $f$ for the other child.  Checking for compatibility requires $O(w)$ time.  This gives the algorithm a total runtime of $O(w{c^*}^{3w}{d^*}^{2w}n)$

\end{proof}


\section{Conclusion}\label{sec:conclude}

Naturally we wonder whether the techniques developed in this paper can be extended.  Most immediately, can we accommodate arbitrary integral or bounded real-valued demands and capacities?  Alternatively, can we extend the solution to the weighted case?  The main challenge to these extensions is that our technique barters capacity for weight in unit quantities.  While we can guarantee that there is some underfull vertex in the graph, we cannot guarantee a bound on the weight of this vertex.  One idea to address this is to design a smarter smoothing process that considers the weights of the vertices chosen in the search.  Alternatively a careful rounding strategy may lead to a bicriteria result in which either the capacities can be exceeded or the demands can be under-covered by a small (bounded) amount.  

Additionally, we consider whether we can extend these techniques to similar problems.  One such problem is capacitated $r$-hop domination, in which a vertex can cover any vertex within a distance of $r$.  This generalizes CDP (in which $r=1$).  The PTAS presented in Section~\ref{sec:ptas} seems to extend to $r$-hop domination by widening the size of our patches to $2(r+1)$ and where $G'$ of Lemma~\ref{lem:improve} has an edge $(u,v)$ for every pair of vertices $u,v$ within a hop-distance of $r$ in $G$.  However, the dynamic program given in Section~\ref{sec:dp} does not readily extend to $r$-hop domination: demand of a non-boundary vertex in a cluster may be covered by a vertex outside of the cluster and encoding all possible ways in which demand can traverse the cluster boundary seems to generate too many configurations.

We hope to generalize the techniques introduced here to apply to a broader range of (hard)-capacitated problems and other problems that forbid multiplicities.

\subsection*{Acknowledgements}
Research supported by NSF Grant CCF-1409520.
  
\bibliographystyle{plain}
\bibliography{main}

\end{document}